\documentclass[%
pra,
 reprint,
 amsmath,amssymb,
 aps,
]{revtex4-2}

\usepackage{latexsym}
\usepackage{acronym}
\usepackage{graphicx}
\usepackage{physics}
\usepackage{amsmath}
\usepackage{amsfonts}
\usepackage{amsthm}
\usepackage{siunitx}
\usepackage{bbm}
\usepackage{bm}
\usepackage[title]{appendix}

\usepackage{tikz}
\usetikzlibrary{arrows,automata}
\usepackage{caption, subcaption}

\newtheorem{theorem}{Theorem}
\newtheorem{lemma}{Lemma}
\newcommand{\hilb}[0]{%
  \mathcal{H}%
}

\newcommand{\tensp}[1]{%
  \mathbin{\mathop{\otimes}\displaylimits_{#1}}%
}

\newacro{POVM}{\textit{positive operator-valued measurement}}

\begin{document}


\title{Quantum Walks can Unitarily Represent Random Walks on Finite Graphs}


\author{Matheus Guedes de Andrade}
\email{mguedesdeand@umass.edu}
\affiliation{%
    Manning College of Information and Computer Science \\
    University of Massachusetts Amherst
}%

\author{Franklin Marquezino}
\email{franklin@cos.ufrj.br}
\author{Daniel Ratton}%
 \email{daniel@cos.ufrj.br}
\affiliation{%
    Department of Computer and Systems Engineering - COPPE\\
    Federal University of Rio de Janeiro
}%

\date{\today}

\begin{abstract}
Quantum and random walks have been shown to be equivalent in the following sense: a time-dependent random walk can be constructed such that its vertex distribution at all time instants is identical to the vertex distribution of any discrete-time coined quantum walk on a finite graph. This equivalence establishes a deep connection between the two processes, far stronger than simply considering quantum walks as quantum analogues of classical random walks. The present work strengthens this connection by providing a construction that establishes this equivalence in the reverse direction: a unitary time-dependent quantum walk can be constructed such that its vertex distribution is identical to the vertex distribution of any random walk on a finite graph at all time instants. The construction shown here describes a quantum walk that matches a random walk without measurements at all time steps (an otherwise trivial statement): measurement is performed in a quantum walk that evolved unitarily until a given time $t$ such that its vertex distribution is identical to the random walk at time $t$. The construction procedure is general, covering both homogeneous and non-homogeneous random walks. For homogeneous random walks, unitary evolution implies time dependency for the quantum walk, since homogeneous quantum walks do not converge under arbitrary initial conditions, while a broad class of random walks does. Thus, the absence of convergence demonstrated for quantum walks in its debut comes from both time-homogeneity and unitarity, rather than unitarity alone, and our results shed light on the power of quantum walks to generate samples for arbitrary probability distributions. Finally, the construction here proposed is used to simulate quantum walks that match uniform random walks on the cycle and the torus.
\end{abstract}

\maketitle

\section{Introduction}

Quantum walks on graphs are one of the few known design techniques for quantum algorithms. Indeed, they have been applied to build quantum algorithms with significant speedups in applications such as searching~\cite{childs2004spatial, meyer2015connectivity, magniez2011search}, Monte Carlo methods~\cite{montanaro2015quantum}, claw finding in graphs~\cite{tani2009claw} and backtracking algorithms~\cite{montanaro2015bakctrack}. Albeit their simple description, quantum walks are extremely powerful since they are universal for quantum computing \cite{childs2009universal, childs2013universal, lovett2010universal}. Intuitively, quantum walks on graphs extend the definition of classical random walks to unitary processes on a Hilbert space that represents graph structure.

Any discrete time coined quantum walk induces a sequence of probability distributions among the vertices of the graph over time and its connection to random walks has been investigated for the infinite integer line~\cite{montero2017quantum}. It has been demonstrated that both quantum and random walks can generate any probability distribution sequence on the integers that respect locality. More precisely, both walks can generate any sequence that assigns, at time $t$, non-zero probability to integers that are either the same, successors, or predecessors of the integers assigned non-zero probability at time $t - 1$.
These connections have also been investigated in the context of continuous-time chiral quantum walks, with a focus on regular structures~\cite{mulken2011continuous, frigerio2021generalized}. The graph random walk Laplacian provided properties that Hamiltonians must exhibit for a quantum system to serve as a quantum counterpart of a classical random walk~\cite{frigerio2021generalized}. It was shown that, for a given random walk dynamics, there are infinitely many compatible Hamiltonian matrices~\cite{frigerio2021generalized}.

While there are distinct ways of defining quantum walks each with different implications, this work considers discrete-time coined quantum walks on finite graphs, a quantum walk model that has been widely adopted~\cite{aharonov1993quantum, konno2018partition, portugal2016staggered, szegedy2004walks}.
Recent work has shown that, regardless of the disparities between quantum and random walks, it is possible to construct a time-dependent random walk on the same graph such that its sequence of vertex probability distribution is identical to that of any given discrete-time coined quantum walk~\cite{andrade2020equivalence}.
When the opposite direction is considered, it is a trivial result that measuring a quantum walk at each time instant recovers the behavior of a random walk. If a quantum walk is measured at time $t$, undergoes one step of unitary evolution, and is measured once again at time $t + 1$, its measurement probabilities evolve following the action of a stochastic matrix. This behavior comes from the destruction of quantum interference caused by the collapse of the wavefunction after the measurement \cite{portugal2013walk}. As an example, it is well known that the Hadamard quantum walk on a cycle collapses to the classic uniform random walk on the same cycle (i.e., a random walk that moves left or right with probability $1/2$), if the walker state is measured at every instant.

Nonetheless, can quantum walk operators be constructed to match classic random walks such that the walker system does not have to be measured at every step? This would require to embed the sequence of vertex probabilities of the classic random walk into the the unitary evolution of a quantum walk. Establishing the equivalence between the coherent evolution of quantum walks and the local evolution of random walks contributes to indicate a fundamental connection between the two processes, which was initially demonstrated in \cite{andrade2020equivalence}. 

In this context, the main contribution of the present work is to describe a rigorous procedure to construct a time-dependent unitary discrete-time coined quantum walk that when measured at time $t$ has the same vertex probability distribution of a classic random walk on the same graph at time $t$. Quantum walks with time-dependent coins exhibit diverse vertex probability distribution sequences based on coin dynamics~\cite{panahiyan2018controlling}, and we use time-dependency to reconstruct the stochastic behavior of the random walk with unitary quantum walks. To illustrate, the procedure here proposed is used to build quantum walks that have identical vertex probability sequences as classic random walks on the cycle and torus graphs. Numerical experiments using these quantum walks and comparison to random walks support the theoretical contributions. 

Note that the theoretical contribution of this work establishes that the divergence between quantum and random walks are not caused by unitarity alone, but are also dependent on time-homogeneity of both processes \cite{aharonov2001quantum}. When time-dependent coins are allowed to be used, the vertex probability distribution of the quantum walk can converge to a stationary distribution, mimicking the temporal dynamics of classic random walks. 

The remainder of this article is structured as follows. In Section \ref{sec:background}, a description of quantum and random walks is given. The algorithm that constructs a quantum walk matching any given random walk, which is the main contribution of this article, is demonstrated in Section \ref{sec:equivalence}. The results of simulation experiments using quantum walks to recover uniform random walks in the cycle and in the torus appear in Section \ref{sec:simulation}. The article is concluded in Section \ref{sec:conclusion}.

\section{Background}\label{sec:background}

Throughout this work, we consider a graph $G = (V, E)$ to be the directed version of an undirected graph, such that, for an edge $(v, u) \in E$, $(u, v) \in E$. In addition, we refer to the set of neighbors of a vertex $v$ as $N(v)$ and the degree of $v$ as $d(v) = |N(v)|$. Essentially, the notation adopted follows the one defined in \cite{andrade2020equivalence}.

It is possible to define a random walk on a graph $G$ as a probability distribution over the vertices of the graph that varies on time depending on edge connectivity, codifying vertex position as a random variable. Consider a probability vector $\pi(t) \in \mathbb{R}_{+}^{|V|}$, such that $\pi_v(t)$ denotes the probability of such random variable to be vertex $v$ at time $t$. A random walk is any process described by the equation
\begin{align}
    & \pi(t + 1) = P(t)\pi(t), \label{eq:rw}
\end{align}
where $P(t)$ is a stochastic matrix containing transition probabilities between vertices that respects the adjacency matrix of $G$, i.e., the entry $p_{vu}(t)$ is the probability to move from $u$ to $v$ at instant $t$ and $p_{vu} > 0$ if, and only if $(u, v) \in E$. Systems described by Equation \ref{eq:rw} are said to be Markovian since the probabilities at a given instant $t$ are completely determined by the probabilities at $t - 1$. Precisely, the evolution of probability for a particular vertex $v$ at instant $t$ in a random walk is a linear combination of the probabilities for the neighbors of $v$ at $t - 1$ following
\begin{align}
    & \pi_v(t) = \sum_{u \in N(v)} p_{vu}(t - 1)\pi_u(t - 1) \label{eq:rw_locality}, \\
    & \sum_{v \in N(u)} p_{vu}(t) = 1 \text{, for every } u \in V,
\end{align}
what implies that $P(t)$ is column stochastic, i.e., its columns sum to one for every instant $t$.

A discrete-time coined quantum walk on a graph $G$ defines the evolution of a unit vector in a Hilbert space $\hilb_w$ that codifies the edges of $G$ \cite{portugal2016coined}. Let $\hilb_v$ and $\hilb_c$ denote Hilbert spaces with dimension $|V|$ and $d_{max} = \max_v d(v)$, respectively. The space $\hilb_w \subseteq \hilb_v \tensp{} \hilb_c$ is spanned by unit vectors $\ket{u, c}$, where $u \in V$ and $c \in \{0, \ldots, d(u) - 1\}$, that can be mapped to edges of the graph through a function $\eta: V \cross C \to V$, with $C = \{0, \ldots, d_{max} - 1\}$. The space $\hilb_v$ is the vertex space of the walker system, codifying elements of $V$ as basis states, while $\hilb_c$ is the coin space of the walker, codifying the degrees of freedom for the walker's movement. Essentially, the wavefunction of the quantum walker at a given time step is a superposition of edges of the graph, having form given by the equation
\begin{align}
    & \ket{\Psi(t)} = \sum_{u \in V, c \in C_v} \psi(u, c, t) \ket{u, c},
\end{align}
where $C_u = \{0,..., d(u) - 1\}$ is the set of degrees of freedom of vertex $u$.

The evolution of the walker state at a discrete time instant $t$ is performed by two time-dependent unitary operators $S(t) : \hilb_w \rightarrow \hilb_w$ and $W(t) : \hilb_w \rightarrow \hilb_w$ on the system state vector as
\begin{align}
    & \ket{\Psi(t + 1)} = S(t)W(t) \ket{\Psi(t)}.\label{eq:qwalk}
\end{align}
$S(t)$ is known as the shift operator and performs a permutation between the edges of the graph that is only allowed to map a given edge $(u, v)$ to an edge $(v, w)$. $W(t)$ is named the coin operator, acting exclusively on $\hilb_C$ by mixing the wavefunction incident to a given edge $(u, v)$ to edges $(u, w)$. Formally, a generic coin operator is defined as 
\begin{align}
    & W(t) = \sum_{v \in V}\dyad{v}{v} \tensp{} W_v(t), \label{eq:coin_op}
\end{align}
where $W_v:\hilb_{C_{v}} \to \hilb_{C_{v}}$ is a unitary operator, where $\hilb_{C_{v}} \subseteq \hilb_C$ is the Hilbert space codifying the degrees of freedom of $v$. In order to define a generic shift operator precisely, it is necessary to define the function $\eta$ used to map edges to states, as well as two auxiliary functions $\sigma: V \cross V \rightarrow C$ and $\sigma^{-1}: V \cross V \rightarrow C$ that respectively map an inward edge of a vertex to an outward edge of vertex and an outward edge of a vertex to its inward correspondent, that are all depicted in Figure~\ref{fig:shift_func}. Thus, consider the meaning of $\eta (v, c) = u$ to be that $u$ is the $c$-th neighbor of $v$; $\sigma(u, v) = c$ meaning that edge $(u, v)$ is mapped to the state $\ket{v, c}$; and that $\sigma^{-1}{(u, v)} = c$ meaning that state $\ket{w, c}$ is mapped to $(u, v)$, whichever vertex $w \in N(u)$ satisfies $\eta(w, u) = c$. Then, the action of the shift operator is expressed as 
\begin{align}
    & \ket{v,c} \rightarrow \ket{\eta(v,c), \sigma(v, \eta(v,c))}.\label{eq:shift_def}
\end{align}

The probability of finding the quantum walker system at a particular state in a given time instant $t$ is characterized by projective measurements of the walker's state. The set of projectors $\{\dyad{v, c}\}$ yields the edge probability distribution
\begin{align}
    & \rho(u, c, t) = \abs{\psi(u, c, t)}^2. \label{eq:rho}
\end{align}
The probability distribution for vertices arises from a set of \ac{POVM} elements $\{E_u\}$ described as
\begin{align}
    & E_u = \sum_{c = 0}^{d_{max} - 1} \dyad{u, c},\label{eq:povms}
\end{align}
which defines the probability of finding the walker in a particular vertex to be
\begin{align}
    & \mu(u, t) = \sum_{c \in C_u} \rho(u, c, t).\label{eq:v_prob}
\end{align}

\begin{figure}
    \centering
    \begin{tikzpicture}[->,>=stealth',shorten >=1pt,auto,node distance=3.8cm, semithick]
        \tikzstyle{every state}=[fill=white, text=black]
        
        \node[state] (A) {$u$};
        \node[state] (B) [left of=A] {$v$};
        \node[state] (C) [right of=A] {$u'$};
        
        \path (A) edge [bend left] node {} (B) 
                edge [bend left] node {$\sigma(v, u)$} (C)
            (B) edge [bend left] node {$\eta(v, c) = \sigma^{-1}(u, u')$} (A)
            (C) edge [bend left] node {} (A); 
    \end{tikzpicture}
    \caption{Visual depiction of auxiliary functions $\eta$ and $\sigma$. $\eta$ gives an ordering for the neighbors of $v$ such that, in this case, $\eta(v, c) = u$. $\sigma$ maps the state $\ket{v, c}$ (edge $(v, u)$) to the state $\ket{u, \sigma(v, u)}$ (edge $(u, u')$). The inverse association $\sigma^{-1}$ connects the state $\ket{u,\sigma(v, u)}$ (edge $(u, u')$) with state $\ket{v, c}$ (edge $(v, u)$).}
    \label{fig:shift_func}
\end{figure}

\section{Quantum walks as non-homogeneous random walks}\label{sec:equivalence}



The unitarity of the quantum walk operators was central to guide the construction of the time-dependent random walk matching a given quantum walk \cite{andrade2020equivalence}. Such property also has fundamental implications that are to be explored in order to design quantum walks capable of coherently matching random walks. Quantum walks evolve with successive applications of unitary operators and lie in a complex unit-radius hyper-sphere in $\hilb_W$, with dimension $|E|$. The vectors describing random walks lie on the positive simplex of dimension $|V|$, which is closed under applications of the stochastic matrices that define the random walk evolution. Thus, the task at hand is two-folded: to represent the probability vectors of random walks as state vectors of $\hilb_W$ by creating a map between $\mathbb{R}^{|V|}$ and $\hilb_W$; and to map the application of a generic stochastic matrix $P(t):\mathbb{R}^{|V|} \to \mathbb{R}^ {|V|}$ to a unitary operator $S(t)W(t): \hilb_W \to \hilb_W$.

Since $\hilb_W$ is a complex space and $|E| \geq |V| - 1$, there are infinitely many mappings that serve the task at hand. Precisely, let $\pi : \mathbb{N} \to [0, 1]^{|V|}$ denote the probability vector of a random walk, at instant $t$, as defined in Section \ref{sec:background}. Any state vector $\ket{\Psi(t)}$ for which
\begin{align}
    \sum_{c \in C_v} \abs{\braket{v, c}{\Psi(t)}}^{2} = \pi_v(t) \label{eq:norm_cond}
\end{align}
is a proper quantum state that represents $\pi(t)$, and thus mimic the evolution of the random walk. The infinitely many possibilities to choose a state for which \eqref{eq:norm_cond} holds entail that the representation of the random walk state is a matter of choice.

Furthermore, it is intuitive that there exists a time-dependent unitary operator $Q(t): \hilb_W \to \hilb_W$ capable of matching the time evolution of states respecting Equation \ref{eq:norm_cond} such that
\begin{align}
    & \ket{\Psi(t+1)} = Q(t)\ket{\Psi(t)}, \label{eq:evo_q}
\end{align}
for every instant $t$. Unitary operators are norm preserving and any unit vector in a Hilbert space can be mapped to any other unit vector in the same space through the application of a unitary operator. Since all state vectors compliant with \eqref{eq:norm_cond} are unitary, there has to exist at least one unitary transformation $Q(t)$ satisfying \eqref{eq:evo_q}. Following this argument, it is necessary to show that there exists an operator satisfying
\begin{align}
    & Q(t) = S(t)W(t),
\end{align}
for every instant $t$, where $S$ and $W$ respect the graph under consideration.

\subsection{Representation of random walks as quantum walker systems}

Initially, it is necessary to define $\ket{\Psi(t)}$. A good representation for the system state is one that simplifies the search for the operators $S(t)$ and $W(t)$. To give intuition on the state representation chosen, consider that the system state is
\begin{align}
    & \ket{\Psi(t)} = \sum_{u \in V, c \in C_u} g(u, c, t) e^{i\theta(u,c,t)} \sqrt{\pi_u(t)} \ket{u,c} \label{eq:rw_state},
\end{align}
where $\theta(u, c, t)$ is an arbitrary complex phase and $g$ respects
\begin{align}
    & \sum_{c \in C_u} g(u, c, t)^2 = 1 \text{ for every } u \in V. \label{eq:g_condition}
\end{align}
The compliance of the equations above with Equation \ref{eq:norm_cond} is a direct consequence of the definition of $\mu$ (Equation \ref{eq:v_prob}). Equation \ref{eq:rw_state} implies that the value of $\rho(u, c, t) = g(u, c, t)^2\pi_u(t)$, while Equation \ref{eq:g_condition} assures that the sum over all degrees of freedom of $u$ yields $\mu(u, t) = \pi_u(t)$.


The diffusion behavior of the quantum walk leads to the definition of $g$ and $\theta$. Note that $W$ acts by mixing the wavefunction among the edges of a vertex and $S$ creates its flow. Consider a particular instant $t$ of the quantum walk with $W(t) = I$ (the identity matrix), such that the mixing behavior is ``turned off" for $t$ and $Q(t)=S(t)$ acts only by creating the flow of the wavefunction. Assuming that $S(t)$ is any valid shift operator and $\eta(u, c) = v$, an inspection of Equation \ref{eq:rw_locality} indicates that a natural choice for the function $g(u, c, t)$ is 
\begin{align}
    & g(u,c,t) = \sqrt{p_{vu}(t)}. \label{eq:g_func}
\end{align}
This choice satisfies Equation \ref{eq:g_condition}, while simultaneously implying that $\mu(u, t) = \pi_u(t)$ and $\mu(u, t+1) = \pi_u(t + 1)$ for all $u \in V$. The first two properties stem directly from the law of total probability. The condition for $t + 1$ comes from the fact that, regardless of the functions $\eta$ and $\sigma$ chosen to define $S(t)$, all states $\ket{v, c}$ have an incident wavefunction that yields a proper proportion of the probabilities of the neighbors of $v$ for every vertex $v \in V$ at instant $t$, such that
\begin{align}
    & \mu(v, t + 1) = \sum_{u \in N^{-}(v)} p_{vu}(t) \pi_u(t).
\end{align}

The given representation is powerful because it is a valid unitary representation that describes the operator $Q(t)$ as a proper quantum walk operator in the particular scenario considered. This representation gives the intuition for a state representation compatible with the process for every $t$. If the result of $W(t)\ket{\Psi(t)}$ is given by the right-hand side of Equation \ref{eq:rw_state} for all instant $t$, $Q(t)$ is properly decomposed into $S(t)W(t)$ for any valid shift operator.

The key aspect is to define the state of the system based on the probabilities of instant $t - 1$, instead of using the probabilities of instant $t$. It is known from the random walk description that the probability of a vertex at instant $t$ is a linear combination of the probability of its neighbors at instant $t - 1$. Thus, the following Lemma inspired by Equation \ref{eq:rw_locality} formalizes the state representation of choice.

\begin{lemma}[Quantum representation of random walks]\label{l:qrw}
Let $P(t)$ be a stochastic matrix that defines a random walk on $G$ such that $\pi(t + 1) = P(t)\pi(t)$. Let $\sigma^{-1}:V \cross V \to C$ be any function that associates an inward edge of a vertex to one of its outward edges defining a valid shift operator for $G$. For $t > 0$, The probability vector $\pi(t)$ can be represented by a discrete-time coined quantum walk state 
\begin{align}
    & \ket{\Psi(t + 1)} = \sum_{v \in V, c \in C_v} e^{i\theta(u, \sigma^{-1}(u, v),t)} \sqrt{p_{vu}(t)\pi_u(t)} \ket{v,c}.
\end{align}
defined on $G$, such that $\theta(u, c, t)$ is an arbitrary complex phase, for all $u \in V, c \in C_u$, $t \in \{0, 1, \ldots\}$, and with
\begin{align}
    & \ket{\Psi(0)} = \sum_{v \in V, c \in C_v} \sqrt{\frac{\pi_v(0)}{d(v)}} \ket{v, c}.
\end{align}
\end{lemma}

\begin{proof}
The proof for $t = 0$ is trivial. Hence, it suffices to show that, for $t >0$, $\ket{\Psi(t)}$ is a unitary vector and that $\mu(v, t) = \pi_v(t)$, for all $v \in V$. The measurement of state $\ket{v, c}$ yields that
\begin{align}
    & \rho(v, c, t + 1) = \abs{e^{i\theta(u, \sigma^{-1}(u, c), t)} \sqrt{p_{vu}(t)\pi_u(t)}}^2.    
\end{align}
Since $p_{vu}(t)$ and $\pi_u(t)$ are positive reals, $\rho(v, c, t) = p_{vu}(t)\pi_u(t)$. The definition of $\mu(v, t)$ gives
\begin{align}
    & \mu(v, t + 1) = \sum_{u \in N^{-}(v)}p_{vu}(t)\pi_u(t)
\end{align}
Equation \ref{eq:rw_locality} implies that $\mu(v, t) = \pi_v(t)$, for all $v \in V$. Since $\norm{\ket{\Psi(t)}} = \sum_{v \in V}\mu(v, t)$, $\ket{\Psi(t)}$ is clearly unitary and the claim is proved.
\end{proof}

\subsection{Complete description of time-evolution}


From Lemma \ref{l:qrw}, any valid shift operator for $G$ can be used to determine the edge maps $\sigma$ and $\sigma^{-1}$ to represent the random walk state. In order to simplify both the analysis and the notation used, consider the following shift operator $S_{RW}$, defined in terms of its auxiliary functions (see Section \ref{sec:background}). Let $\eta:V \cross C \to V$ be defined such that the $c\text{-th}$ neighbor of $v$ is the neighbor of $v$ with the $c\text{-th}$ smallest label. Formally, for all $u\in V$, all $c,c' \in C_u$, $c \neq c'$, it holds
$$\eta(u, c) < \eta(u, c') \iff c < c'.$$
Note that, for each $(u, v) \in E$, there exists a pair $c \in C_u$, $c' \in C_v$ such that $\eta(u, c) = v$ and $\eta(v, c') = u$. Thus, let $\sigma(u, v) = c'$ and $\sigma(v, u) = c$. Furthermore, let $\sigma^{-1}(u,v) = \sigma(v ,u)$ and $\sigma^{-1}(v,u) = \sigma(u, v)$. Precisely, $S_{RW}$ is the flip-flop shift operator that maps $(u, v)$ to $(v, u)$ and is well-defined for any graph $G$ of interest. The definitions of $\eta$ and $\sigma$ for $S_{RW}$ will be used throughout this section.

The $S_{RW}$ operator yields describing the state of a vertex $v$ at time $t$ as the vector
\begin{align}
    & \ket{\Psi(v, t + 1)} = \sum_{u \in N^{+}} e^{i\theta(u, c,t)} \sqrt{p_{vu}(t)\pi_u(t)}\ket{v, c'}\label{eq:psi},
\end{align}
where the dependency of $c$ and $c'$ on $u$ and $v$ is omitted, i.e $\eta(u, c) = v$ and $\eta(v, c') = u$. From the analysis of the state representation on the previous section, it is enough to ensure that, for every instant $t$, the action of $W(t)$ maps
\begin{align}
    & \ket{\Psi(v, t)} \to \ket{\Phi(v, t)},
\end{align}
where
\begin{align}
    & \ket{\Phi(v, t)} = \sum_{c' \in C_v} e^{i\theta(v,c',t)} \sqrt{p_{uv}(t) \pi_v(t)} \ket{v, c'} \label{eq:phi}.
\end{align}
The definition of the coin operator (Equation $\ref{eq:coin_op}$) implies that each vertex $v$ has its own independent unitary mixing behavior $W_v$. In addition, it is a well known result from linear algebra that any operator that changes orthonormal basis is unitary. Thus, the following Lemmas respectively provide formal constructions for a set of linearly independent vectors on the coin subspace of a vertex and the coin operator $W(t)$ itself.

\begin{lemma}[Linear independent set construction]\label{l:li}
Let $\hilb_{w, d(v)} \subset \hilb_W$ be the subspace that represents the coin space $\hilb_{d(v)}$ of a vertex $v$. Let $\ket{a} \in \hilb_{w, d(v)}$ be any vector in the subspace. Let $\beta$ be the basis of edges $\{\ket{u,c}\}$ for $\hilb_{w, d(v)}$. The set
\begin{align}
A = \{\ket{a}\} \cup \zeta(a, v) \cup B\}
\end{align}
is linearly independent, where $\zeta(a, v) = \{\ket{v, c} : \braket{a}{v, c} = 0\}$ and $B \subset \beta \setminus \zeta(a, v)$ is any subset of $\beta \setminus \zeta(a, v)$ with cardinality $|B| = |\beta \setminus \zeta(a, v)| - 1$.
\end{lemma}
\begin{proof}
It is clear that all vectors from $A \setminus \{\ket{a}\}$ are orthogonal, since they are a subset of the basis $\beta$. In the case where $\zeta(a, v) = \{ \}$, $\braket{a}{v, c} \ne 0$ for all $\ket{v, c} \in \beta$, what implies that exists a $c'$ such that $\braket{a}{v, c'} > 0$ while $\braket{v, c}{v, c'} = 0$ for all $\ket{v, c} \in B$. Hence, it is impossible to write $\ket{a}$ as a linear combination of vectors in $B$, and $A$ is a set of linearly independent vectors.

In the case where $\zeta(a, v) \neq \{ \}$ the construction of $A$ implies that the condition of the existence of $c'$ holds because the vectors from $\zeta(a, v)$ are orthogonal to $\ket{a}$ and exactly one of the vectors of the set $\beta \setminus B$ is not a member of $B \cup \zeta(a, v)$.
\end{proof}

\begin{lemma}[Coin operators for random walks]\label{l:unit}
Let $\hilb_{d(v)}$ denote the Hilbert space defined by the degrees of freedom of a vertex $v$. Let the sets of vectors $\alpha$ and $\beta$ be two orthonormal basis for $\hilb_{d(v)}$, where $\alpha_k$ and $\beta_k$ are, respectively, the $k\text{-th}$ vectors of $\alpha$ and $\beta$. Let 
\begin{align}
    & \alpha_0 = \frac{1}{\sqrt{\braket{\Phi(v, t)}}} \ket{\Phi(v, t)}, \\
    & \beta_0 = \frac{1}{\sqrt{\braket{\Psi(v, t)}}} \ket{\Psi(v, t)},
\end{align}
with $\Psi(v, t)$ and $\Phi(v, t)$ given by Equations \ref{eq:psi} and \ref{eq:phi}, respectively. The operator
\begin{align}
    & W_v(t) = \sum_{k = 0}^{d(v) - 1} \dyad{\alpha_{k}}{\beta_{k}}, \label{eq:vertex_coin_op}
\end{align}
is unitary, inducing a unitary operator $W(t) = \sum_{v \in V} \dyad{v}{v} \bigotimes W_v(t)$ on $\hilb_w$.
\end{lemma}
\begin{proof}
Note that $\braket{\Psi(u,t)} = \braket{\Phi(u,t)}$. It follows trivially from the completeness relation that
$$W_v^{\dagger}(t)W_v(t) = W_v(t)W_v^{\dagger}(t) = I,$$
and $W(t)$ is unitary.
\end{proof}

Finally, the following Theorem states that, for any given random walk, a statistically equivalent quantum walk in terms of vertex probabilities can be constructed assuming time-dependent coin operators.

\begin{theorem}\label{th:random_walk_unitarity}
Let $P(t)$ be a stochastic matrix that defines the evolution of a random walk on a graph $G$, such that, for all $t$, $\pi(t + 1) = P(t)\pi(t)$. For every instant $t$, the quantum walk with state $\ket{\Psi(t)}$ given by Lemma \ref{l:qrw}, with fixed shift operator $S(t) = S_{RW}$ and coin operator $W(t)$ given by Lemma $\ref{l:unit}$, evolves according to
\begin{align}
    & \ket{\Psi(t + 1)} = S_{RW}W(t)\ket{\Psi(t)},
\end{align}
such that $\mu(u, t) = \pi_u(t)$ and $\mu(u, t + 1) = \pi_u(t + 1)$ for all $u \in V$.
\end{theorem}

\begin{proof}
For every $u \in V$, the conditions for $\mu(u, t) = \pi_u(t)$ and $\mu(u, t+1) = \pi_u(t + 1)$ are ensured by Lemma \ref{l:qrw}. As a valid shift operator, $S_{RW}$ is unitary. At instant $t$, construct two sets of linearly independent vectors $A_{v,1}$ and $A_{v, 2}$ by respectively applying Lemma \ref{l:li} to $\ket{\Phi(v, t)}$ and $\ket{\Psi(v,t)}$ for every vertex $v$. Use the Gram-Schmidt procedure on the sets $A_{1,v}$ and $A_{2, v}$ to generate the orthonormal basis $\alpha_{k,v}$ and $\beta_{k, v}$ respectively. Take $W(t)$ as the unitary operator defined by Lemma \ref{l:unit} using all basis $\alpha_{v, k}$ and $\beta_{v, k}$. Hence, $S_{RW}$ and $W(t)$ are unitary and well defined for every instant $t$ and the claim is proved.
\end{proof}

It is essential to note that the procedure to construct the quantum walk of Theorem \ref{th:random_walk_unitarity} is not unique. In addition to the infinite possibilities of representation that lead to distinct definitions for $S(t)$ and $W(t)$, the operator $W(t)$ can also be defined differently. In fact, the Gram-Schmidt procedure is just one convenient way to define $W(t)$. Nonetheless, there may exist alternative definitions that could be more efficient under specific conditions, such as particular random walk definitions and graphs.

Together with the procedure that construct random walks matching quantum walks \cite{andrade2020equivalence}, Theorem \ref{th:random_walk_unitarity} reveals that unitary discrete-time coined quantum walks and non-homogeneous random walks are intrinsically related. Knowing the time-dependent stochastic matrix $P(t)$ and the probability vector $\pi(t)$ allows for the construction of a quantum walk operator $S(t)W(t)$ and the state vector $\ket{\Psi(t)}$, and vice-versa.


\section{Convergent quantum walks on graphs}\label{sec:simulation}
Theorem \ref{th:random_walk_unitarity} considers general, non-homogeneous random walks. Therefore, the theorem can be used to construct statistically equivalent quantum walks for the particular case of time-homogeneous random walks. In this case, the convergence of the vertex probability vector $\pi(t)$ for arbitrary initial conditions is assured when the random walk is irreducible and aperiodic. Despite the absence of convergence for the wavefunction caused by unitarity~\cite{aharonov2001quantum}, the convergence of the vertex probability does not harm the construction of the equivalent quantum walk. To illustrate, consider a quantum walk where the wavefunction is permuted among the edges of a vertex perpetually, such that $\Psi(v, c, t) = \Psi(v, c', 0)$ for $t > 0$ and $c' \in C_v$. The vertex probability is the same for all $t$ while the wavefuction keeps alternating forever on the edges, and thus does not converge. Note also that a time-homogeneous stochastic matrix $P(t)=P$ does not implies on a time-independent coin operator $W(t)=W$. 

In this section, we apply Theorem \ref{th:random_walk_unitarity} to obtain quantum walks that replicate uniform random walks. Note that such quantum walks converge to the same stationary distribution of their random walks counterpart. We focus on uniform random walks on cycle and torus graphs, and numerically evaluate the quantum walks obtained from Theorem \ref{th:random_walk_unitarity} to show their convergence. It is worth emphasizing that the equivalence theorem provides a way to simulate random walks through the simulation of their correspondent quantum walks, a technique which is applied in the numerical analysis presented in this section.

Cycle and torus graphs are examples of regular graphs, i.e., graphs which all vertices have the same degree. The cycle graph has vertices with degree two, while the torus has vertices with degree four. The uniform random walk assigns equal transition probabilities to all edges departing from a graph's vertex. Thus, each outgoing edge of a vertex in the cycle and torus is traversed with probability 1/2 and 1/4 by an uniform random walk moving from that vertex, respectively. Uniform random walks on regular graphs have been well studied and there are known conditions for the convergence of probability distributions with time, i.e., irreducibility and aperiodicity, as well as descriptions for its stationary distributions~\cite{haggstrom2002finite}. In particular, if the Markov chain underlying the uniform random walk on a given regular graph is irreducible and aperiodic, the stationary distribution of the random walk is uniform in the vertex set. More precisely, let $M$ denote the column stochastic matrix driving the time evolution of a uniform random walk on an arbitrary regular graph of $n$ vertices. Under the irreducibility and aperiodicity of $M$, it follows that
\begin{align}
    \lim_{t \to \infty}M^{t} \pi(0) = \frac{\vec{1}}{n}, \label{eq:regular_stationary}
\end{align}
where $\vec{1} \in \mathbb{R}^{n}$ is a vector with all entries equal to one and $\pi(0)$ is an arbitrary initial distribution for the random walk. The stochastic matrices of uniform random walks on the torus are irreducible and aperiodic independent of the number of vertices in the graph. In the case of the cycle, the matrices are irreducible for all cases, although are only aperiodic for cycles with an odd number of vertices. 

Quantum walks that replicate the vertex distribution of homogeneous random walks on the cycle and torus converge to the stationary distribution, following~\eqref{eq:regular_stationary}. In particular, let $M_C$ and $M_T$ denote the column stochastic matrices that drive space-homogeneous random walks on the cycle and torus graph of $n$ vertices, respectively. Let $W_C(t)$ and $W_T(t)$ denote the quantum walk operators derived from $M_C$ and $M_T$ from Theorem \ref{th:random_walk_unitarity}, respectively, using an arbitrary initial condition $\pi(0) \in [0, 1]^{n}$ in both cases. $W_C(t)$ and $W_T(t)$ induce vertex probability distribution sequences $\mu_C$ and $\mu_T$ following \eqref{eq:v_prob} for which \eqref{eq:regular_stationary} yields
\begin{align}
    & \lim_{t \to \infty} \mu_C(v, t) = \lim_{t \to \infty} \mu_T(c, v) = \frac{1}{n}, \text{ for all } v.\label{eq:qwregular_stationary}
\end{align}

The probability $\mu(v, t)$ of measuring a quantum walk at instant $t$ in vertex $v$ defined in \eqref{eq:v_prob} is based on the assumption that the walk system undergoes unitary evolution until instant $t$ and its measured instantaneously at $t$. Let $\mathcal{T} = \{0, 1, \ldots, t\}$ denote a set of discrete time steps for which a given quantum walk is to be simulated for. Let $\hat{\mu}_k(v, t)$ denote the empirical vertex probability distribution obtained when $k$ quantum walk systems identically prepared in state $\ket{\psi(0)}$ are measured after $t$ steps of unitary evolution. A $k$-sample, $t$-step quantum walk simulation is, in the scope of this work, the point estimation of the vertex probability distributions $\mathcal{V}^{k} = \{\mathcal{V}_0^{k}, \mathcal{V}_1^{k}, \ldots, \mathcal{V}_t^{k},\}$, with $\mathcal{V}_{j}^{k}: V \to \mathbb{R}^{+}$ such that $\mathcal{V}_{j}^{k}(v) = \hat{\mu}_k(v, j)$ for all $v \in V$, $j \in \{0, 1, \ldots, t\}$. It follows that the number of quantum walk systems that must be prepared to estimate $\mathcal{V}^{k}$ is $tk$.

\begin{figure*}
    \centering
    \includegraphics[scale=0.45]{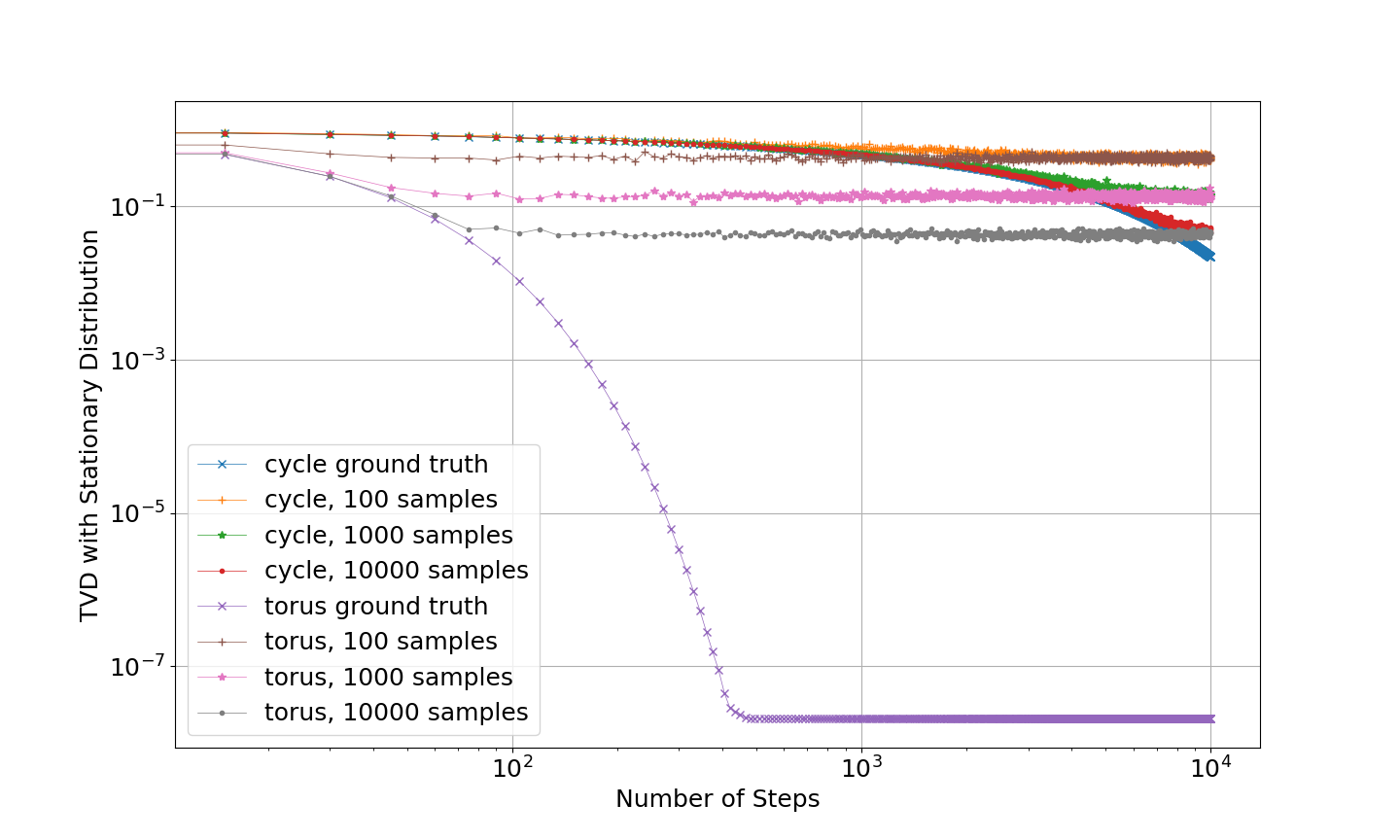}
    \caption{Total variation distance (TVD) with stationary distribution at steps multiples of 15. Ground-truth curves show the TVD between the random walk vertex probability distribution $\pi(t)$ and the stationary distribution $\pi^{*}$, which is uniform distribution for both graphs. Other curves show the TVD between $\mathcal{V}_t^{k}$ and $\pi^{*}$ for $k \in \{10^{2}, 10^{3}, 10^{5}\}$.}
    \label{fig:convergence}
\end{figure*}

Empirical distributions $\mathcal{V}^{k}$ are simulated for a cycle and a torus graph with $|V| = 121$ vertices each, using a discrete delta function for the initial condition $\pi(0)$, and for various values of $k$ and $t$. The number of vertices is selected to be odd in order to ensure that the convergence conditions for the cycle graph are satisfied. Quantum walk operators are obtained by using $\pi(0)$, $M_C$, and $M_T$ in Theorem \ref{th:random_walk_unitarity}.
Since convergence conditions are met, the vertex distribution sequences of the quantum walks approach the uniform distribution $\mu^{*}(v, t) = 1 / 111$ as $t$ approaches the mixing time of the corresponding random walks~\cite{levin2017markov}. Mixing times for uniform random walks on the torus and cycle graphs are known to be $O(|V|)$ and $O(|V|^{2})$, respectively~\cite{levin2017markov}.
Moreover, this initial analysis assesses the convergent behavior of the quantum walks considered through the \textit{Total Variation Distance} (TVD) $D_j$ between the point-estimate quantum walk probability $\hat\mu_{k}$ and the random walk stationary distribution $\mu^{*}$, which assumes the form
\begin{align}
    D_{t'}(\hat{\mu}_k, \mu^{*}) = \frac{1}{2} \sum_{v \in V} \abs{\hat\mu_{k}(v, t') - \frac{1}{|V|}}.
\end{align}
for $t' \in \{0, \ldots, t\}$. Results are reported in FIG.\ref{fig:convergence}, where quantum walks for both the cycle and torus are simulated until time $t = |V|^2$. The TVD of point estimates obtained from simulation are shown, together with ground truth values obtained by numerically evaluating the wavefunction of the quantum walks through Equation \eqref{eq:qwalk} and computing measurement probabilities with the POVMs shown in Equation \eqref{eq:povms} for times $t' \in \{0,\ldots, t\}$. As expected, increasing the number of samples reduces the TVD for all time instants for both graphs. Furthermore, the TVD for the torus approaches zero faster than that of the cycle, which is expected considering their mixing times.

Results in FIG.\ref{fig:convergence} highlight that the variance of the estimator for the vertex distribution increases as the systems approach the stationary state. This phenomena can be explained as follows. The initial distribution for both walks is a discrete delta distribution, which has zero variance. In contrast, their stationary distribution is uniform, which has maximum variance. Thus, the variance of $\pi(t)$ increases as $t$ moves from zero to the mixing time, reaching the maximum possible value at the mixing time itself. Since the walk on the torus approaches the stationary distribution an order of magnitude faster than the walk on the cycle, the variance of $\pi(t)$ on the torus grows faster with $t$ initially. This fact is visible in FIG.\ref{fig:convergence} by noting that the point-estimate curves for the torus start diverging from the ground truth much earlier than the respective curves for the cycle. The time where this divergence occurs indicates the moment where the TVD between the empirical distributions and the stationary distribution becomes dominated by estimation error, rather than the actual difference between $\pi(t)$ and the uniform distribution. Indeed, increasing time beyond this value will not reduce the TVD of the point-estimates. 

\begin{figure*}
    \centering
    \includegraphics[scale=0.45]{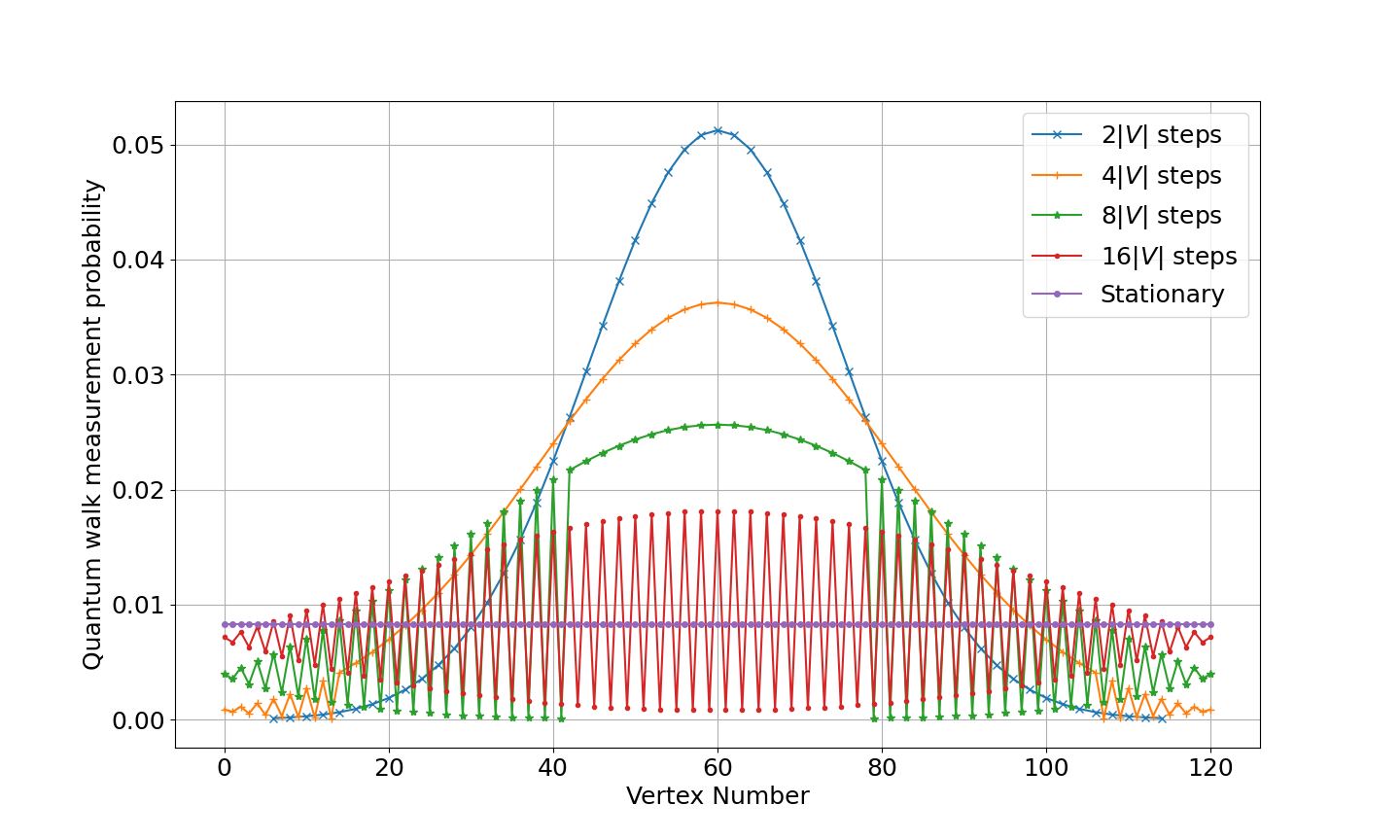}
    \caption{Probability distribution $\mu(\cdot,t)$ for cycle and torus at times $t \in \{2 |V|, 4 |V|, 8 |V|, 16 |V|\}$ and the stationary distribution. The initial condition $\pi(0)$ used is a delta function at vertex labelled 121. For clarity, we omit $\mu(\cdot, t)$ for vertices $v$ with $\mu(v, t) \leq 10^{-4}$.}
    \label{fig:convergence_fixedsamples}
\end{figure*}

In order to illustrate convergence of the quantum walk, FIG.~\ref{fig:convergence_fixedsamples} shows the probability distribution $\mu(\cdot, t)$, Equation~\eqref{eq:v_prob}, for the quantum walk that mimics the uniform random walk on the cycle at times $t \in \{2|V|, 4|V|, 8|V|, 16|V|\}$ along with its stationary distribution. The initial condition $\pi(0)$ considered is a delta function peaking at vertex $60$. The cycle is used for visual purposes: the labelling of the vertices on the cycle allows for a meaningful 2D-representation of vertex probabilities.
For clarity, vertex probability $\mu(v, t)$ smaller than $10^{-4}$ are omitted for all curves, but only odd numbered vertices have such a small probability (since the observation time is even). Note that as time increases, the probability associated with odd numbered vertices increases at even time steps. At time $2|V|$, no odd vertex meets this quota, while at time $16|V|$ all odd vertices have probability above $10^{-4}$. The curves show the shape of $\mu(\cdot, t)$ and the convergence of the vertex distribution as $t$ increases to the steady state distribution, a behavior well-known for the random walk that is now also observed by the corresponding quantum walk.


\section{Discussion} \label{sec:conclusion}

This work provided a methodology to build a time-dependent discrete time coined quantum walk that has a vertex probability distribution over time that is identical to that of any classic random walk on the same graph. In a nutshell, given a finite arbitrary graph, a time-dependent quantum walk can exactly replicate the vertex distribution of a random walk over time. Interestingly, this implies that the vertex probability distributions of time-dependent quantum walks can converge to the steady state distribution of random walks. 

Recent prior work has shown that a time-dependent biased random walk can exhibit the same vertex distribution over time as any discrete time coined quantum walk~\cite{andrade2020equivalence}, a result in the opposite direction. Both results require a time-dependent process (quantum walk or random walk) in order to mimic the evolution of a time homogeneous process (random walk or quantum walk, respectively). In general, such equivalence is not possible without a time-dependent model. Indeed, time-dependence significantly increases the expressiveness of such processes, for both random and quantum walks, allowing one to mimic the other. Considering time-dependent processes, quantum and random walks are two equivalent models to represent the evolution of vertex probability distribution on graphs. 

Last, while the proposed methodology to build time-dependent quantum walks can be applied to any random walk on any finite graph, the solution it generates is not unique: there are other time-dependent quantum walks that can mimic the evolution of random walks. In fact, there are other methodologies to consistently build time-dependent quantum walks capable of replicating the time evolution of random walks in a unitary way. Clearly, different methodologies have different implications (e.g., in terms of quantum circuit complexity) and determining the most efficient methodology is theme of future work.

\bibliography{references}
\bibliographystyle{unsrt}

\end{document}